\documentclass[11pt]{article}
\usepackage[english]{babel}
\usepackage{amsmath,amssymb,amsthm,graphicx}

\usepackage[format=plain,indention=0cm,justification=centerlast,
font=small,labelfont=bf,labelsep=period,width=0.8\textwidth]{caption}

\usepackage[pdfpagemode=UseOutlines,
bookmarks=true, bookmarksopenlevel=3, bookmarksopen=true, bookmarksnumbered=true, unicode=true, pdffitwindow=true,pdfpagelayout=SinglePage,pdfhighlight=/P, colorlinks=true,linkcolor=blue,citecolor=blue,urlcolor=blue]{hyperref}

\advance\textwidth 20mm  \advance\oddsidemargin -10mm
\advance\textheight 30mm \advance\topmargin -18mm

\theoremstyle{plain}
\newtheorem{theorem}{Theorem}
\newtheorem{statement}[theorem]{Statement}
\theoremstyle{definition}
\newtheorem{definition}{Definition}
\newtheorem{remark}{Remark}

\newcommand{\Integer}{\mathbb{Z}}
\newcommand{\const}{\mathop{\rm const}}

\title{Integrable 7-point discrete equations and evolution lattice
equations of order 2}
\author{V.E. Adler\footnote{L.D. Landau Institute for Theoretical Physics,
Ac.~Semenov 1-A, 142432 Chernogolovka, Russian Federation. E-mail: adler@itp.ac.ru}}
\date{27 May 2017}

\begin{document}\thispagestyle{empty}
\maketitle

\begin{abstract}
We consider differential-difference equations that determine the continuous symmetries of discrete equations on the triangular lattice. It is shown that a certain combination of continuous flows can be represented as a scalar evolution lattice equation of order 2. The general scheme is illustrated by a number of examples, including an analog of the elliptic Yamilov lattice equation.
\medskip

\medskip

\end{abstract}

\section{Introduction}\label{s:intro}

One of characteristic features of integrable discrete equations is the existence of continuous symmetries, that is, vector fields $\partial_x$ which preserve the equation $E[u]=0$ on the lattice: $\partial_x(E)=0|_{E=0}$. In particular, it is well known that differential-difference equations define continuous symmetries for quad-equations and their higher order generalizations \cite{Papageorgiou_Nijhoff_1996, Levi_Petrera_Scimiterna_Yamilov_2008, Adler_Postnikov_2014, Garifullin_Mikhailov_Yamilov_2014}. In this paper we study a relation of evolution lattice equation of second order
\begin{equation}\label{ut2}
 u_{,t}=f(u_{-2},u_{-1},u,u_1,u_2),\quad
 u_j=u(t,n+j),\quad u_{,t}=\partial u/\partial t,
\end{equation}
with another type of discrete equations, namely, 7-point equations on the triangular lattice
\begin{equation}\label{i.7-point}
 f(q,q_{1,0})-\tilde f(q,q_{-1,0})+g(q,q_{0,1})-\tilde g(q,q_{0,-1})
  +h(q,q_{-1,-1})-\tilde h(q,q_{1,1})=0.
\end{equation}
Continuous symmetries for equations (\ref{i.7-point}) are determined by vector fields of the form
\begin{equation}\label{i.Dx}
 q_{,x}=a(q,f(q,q_{1,0})-\tilde g(q,q_{0,-1})-\tilde h(q,q_{1,1}))
\end{equation}
and similar equations that arise under reflections of the lattice. Here we do not discuss the problem of classification of all consistent pairs (\ref{i.7-point}), (\ref{i.Dx}) and confine ourselves to examples that can be obtained by starting from already known integrable equations (\ref{i.7-point}). Such equations were studied, for example, in \cite{Suris_1995, Suris_1996, Suris_1997, Adler_2000a, Adler_2000b, Adler_Suris_2004, Boll_Suris_2009}, where the zero curvature representations and B\"acklund transformations were constructed and a continuous limit to differential-difference equations of the Ruijenaars--Toda lattice type was investigated. Continuous symmetries (\ref{i.Dx}) have not been considered before, with the exception of one example \cite{Adler_2000a}.

The main result of the paper is related to the observation that a symmetrized linear combination of flows of the type (\ref{i.Dx}) is described, under certain additional conditions, by a lattice equation of the form (\ref{ut2}). A direct search shows that such equations exist for all examples from the cited works. Almost all of them are new if we distinguish equations up to the point transformations.

For the case $h=\tilde h=0$ corresponding to 5-point equations on the square lattice, the continuous symmetries of type (\ref{i.Dx}) were stidied in \cite{Adler_Suris_2004, Adler_Shabat_2006}. Here,  instead of (\ref{ut2}), the lattice equations of order 1 appear, as it was shown in \cite{Adler_Suris_2004} by example related with the elliptic Yamilov lattice (a semi-discrete analog of the Krichever--Novikov equation)
\[
 u_{,t}=\frac{H(u,u_1)}{u_1-u_{-1}}-\frac{1}{2}\partial_1(H(u,u_1)),
\]
where $H$ is a symmetric biquadratic polynomial, $\partial_i:=\partial/\partial u_i$. Recall, that this equation is the main case in the classification of integrable lattice equations of first order \cite{Yamilov_1983, Yamilov_2006}. The classification problem for second order equations (\ref{ut2}) is much more difficult and in this direction only few particular results were obtained so far \cite{Adler_2014, Adler_2016, Garifullin_Yamilov_Levi_2016}. New equations found in this paper essentially expand the list of examples and give some idea of the complexity of the answers in this problem (although they obviously do not cover all integrable cases; in particular, all the lattice equations under scrutiny should admit zero curvature representations in $2\times2$ matrices, as well as the parent equations (\ref{i.7-point}); on the other hand, such examples as the Bogoyavlensky lattice and some its generalizations are related to $3\times3$ matrices).

Section \ref{s:construction} contains the general definition of consistency of 7-point equations (\ref{i.7-point}) and flows of the type (\ref{i.Dx}), and the derivation scheme of the associated lattice equations (\ref{ut2}). An important intermediate step here is related with 2-component evolution systems corresponding to the choice of dynamical variables on a pair of neighboring lines of the lattice. Such systems form an independent class of equations deserving a separate study.

The examples given in section \ref{s:difference} are related with the list of 7-point Lagrangian equations that are invariant under translation $q\to q+\varepsilon$ \cite{Adler_2000a, Adler_2000b, Boll_Suris_2009}. We will demonstrate that the lattice equations (\ref{ut2}) can be attached to all equations from the list; in the cases when the parent equation (\ref{i.7-point}) is not symmetric with respect to reflections of the lattice, one has to apply an additional point transformation. Two of obtained equations have been considered earlier \cite{Garifullin_Yamilov_2012, Garifullin_Mikhailov_Yamilov_2014, Adler_2016}, the rest ones seem to be new. It turns out, however, that all of them admit difference substitutions $w=\phi[u]$ which bring to the same equation \cite{Adler_2016}
\begin{equation}\label{i.wt}
 w_{,t}=(w+1)\Bigl(w_2\Bigl(\frac{1}{w_1}+1\Bigr)w
   -w\Bigl(\frac{1}{w_{-1}}+1\Bigr)w_{-2} +w_1-w_{-1}\Bigr),
\end{equation}
which plays a role of a kind of universal object in the theory. About the equation (\ref{i.wt}) little is known so far, but, undoubtedly, it deserves a detailed study.

In section \ref{s:Q4} we consider 7-point equations from \cite{Adler_Suris_2004}, which are related with the $Q_4$ quad-equation and its degeneracies (symmetrical examples from section \ref{s:difference} can be obtained from here as limiting cases, as well). This brings us to the most general family of integrable lattice equations in this paper,
\begin{equation}\label{i.HPut}
 u_{,t}=\frac{H(u_{-1},u)H(u,u_1)(u_2-u_{-2})}
          {P(u_{-2},u_{-1},u,u_1)P(u_{-1},u,u_1,u_2)},
\end{equation}
where $P(u_{-1},u,u_1,u_2)$ is an affine-linear polynomial with the symmetry group of a square
\[
 P(u_{-1},u,u_1,u_2)=P(u_2,u_1,u,u_{-1})=P(u_1,u,u_{-1},u_2),
\]
and $H(u,u_1)=H(u_1,u)$ is a symmetric biquadratic polynomial determined through $P$ by relation
\[
 \const H(u,u_1)=\partial_{-1}(P)\partial_2(P)-\partial_{-1}\partial_2(P)P.
\]
The generic equation (\ref{i.HPut}) corresponds to an elliptic curve, and it is likely that it cannot be brought by substitutions to other lattice equations, in the same way as the Yamilov lattice. In the cases of trigonometric and rational degenerations, equation (\ref{i.HPut}) admits a substitution into (\ref{i.wt}).

\section{General scheme}\label{s:construction}

In this section we explain how the lattice equations (\ref{ut2}) appear in the theory of discrete equations on triangular lattice. The presented construction include the following steps:
\begin{itemize}\setlength{\itemsep}{-2pt}
\item[---] derivation of continuous symmetries for discrete equations;
\item[---] an equivalent form of these flows as 2-component lattice equations;
\item[---] rewriting in the scalar form, possible in a symmetric situation.
\end{itemize}

We start from the discrete equations of the form
\begin{equation}\label{7-point}
 f(q,q_{1,0})-\tilde f(q,q_{-1,0})+g(q,q_{0,1})-\tilde g(q,q_{0,-1})
  +h(q,q_{-1,-1})-\tilde h(q,q_{1,1})=0,
\end{equation}
where function $q=q(n_1,n_2)$ depends on two integer variables and subscripts denote the offsets: $q_{i,j}=q(n_1+i,n_2+j)$. We will use also the following notations for the shift operators generating the triangular lattice:
\[
 T_1:q\to q_{1,0},\quad T_2:q\to q_{0,1},\quad T_3=T^{-1}_1T^{-1}_2.
\]
The basic continuous symmetries for equations (\ref{7-point}) are defined by differential-difference equations of the form
\begin{equation}\label{7-Dx}
 q_{,x}=a(q,F),\quad
 F=f(q,q_{1,0})-\tilde g(q,q_{0,-1})-\tilde h(q,q_{1,1}).
\end{equation}
Taking the initial equation into account, one can write the differentiation $\partial_x$ also in the equivalent form
\begin{equation}\label{7-tildeDx}
 q_{,x}=a(q,\tilde F),\quad
 \tilde F=\tilde f(q,q_{-1,0})-g(q,q_{0,1})-h(q,q_{-1,-1}).
\end{equation}

\begin{definition}\label{def:cons}
Discrete equation (\ref{7-point}) and continuous equation (\ref{7-Dx}) are consistent, if the equality $\partial_x(F-\tilde F)|_{F=\tilde F}=0$ holds identically on the lattice.
\end{definition}

From the computational point of view, the check of the compatibility condition amounts to the equality
\begin{multline}\label{DxFF}
 \left(\frac{\partial F}{\partial q}
   +\frac{\partial F}{\partial q_{1,0}}T_1
   +\frac{\partial F}{\partial q_{0,-1}}T^{-1}_2
   +\frac{\partial F}{\partial q_{1,1}}T^{-1}_3\right)(a(q,\tilde F))\\
 -\left(\frac{\partial\tilde F}{\partial q}
   +\frac{\partial\tilde F}{\partial q_{-1,0}}T^{-1}_1
   +\frac{\partial\tilde F}{\partial q_{0,1}}T_2
   +\frac{\partial\tilde F}{\partial q_{-1,-1}}T_3\right)(a(q,F))
 =\phi[q](F-\tilde F)
\end{multline}
with an indefinite factor $\phi$ which may depend, as it is easy to prove, only on the same set of variables which are involved in (\ref{7-point}).

It should be stressed that, although equations (\ref{7-Dx}), (\ref{7-tildeDx}) contain three independent variables $n_1,n_2$ and $x$, this {\em is not} a 3-dimensional system, because we consider both equations simultaneously and this pair of equations is defined only on the two-dimensional equation (\ref{7-point}). (Such a situation is rather typical for discrete and semi-discrete equations, but, unfortunately, it is often misused in the literature. Dimension of a system is one greater than the dimension of submanifold, on which the initial conditions for generic solution are given, and may not coincide with the number of independent variables.)

Equations (\ref{7-Dx}), (\ref{7-tildeDx}) can be rewritten in the form of a differential-difference equation with two components corresponding to two neighboring lines of the lattice parallel to the axis $T_2$ or to the axis $T_3$.

\begin{statement}\label{st:qpx}
Let $j$ be a constant integer. The variables $q(k)=q(j,k)$, $p(k)=q(j+1,k)$ satisfy the system of equations
\begin{align*}
 q_{,x}&=a(q,f(q,p)-\tilde g(q,q_{-1})-\tilde h(q,p_1)),\\
 p_{,x}&=a(p,\tilde f(p,q)-g(p,p_1)-h(p,q_{-1}));
\end{align*}
the variables $q(k)=q(-k,j-k)$, $p(k)=q(-k,j+1-k)$ satisfy the system
\begin{align*}
 q_{,x}&=a(q,f(q,p)-\tilde g(q,p_1)-\tilde h(q,q_{-1})),\\
 p_{,x}&=a(p,\tilde f(p,q)-g(p,q_{-1})-h(p,p_1)).
\end{align*}
\end{statement}

Further on, we will be interested in the situation, when equation (\ref{7-point}) admits, along with $\partial_x$, two another compatible flows of similar kind, corresponding to the $\Integer_3$-symmetry of the lattice (see fig.~\ref{fig:7-point} on the left).

\begin{figure}
\centerline{
\includegraphics[scale=0.6]{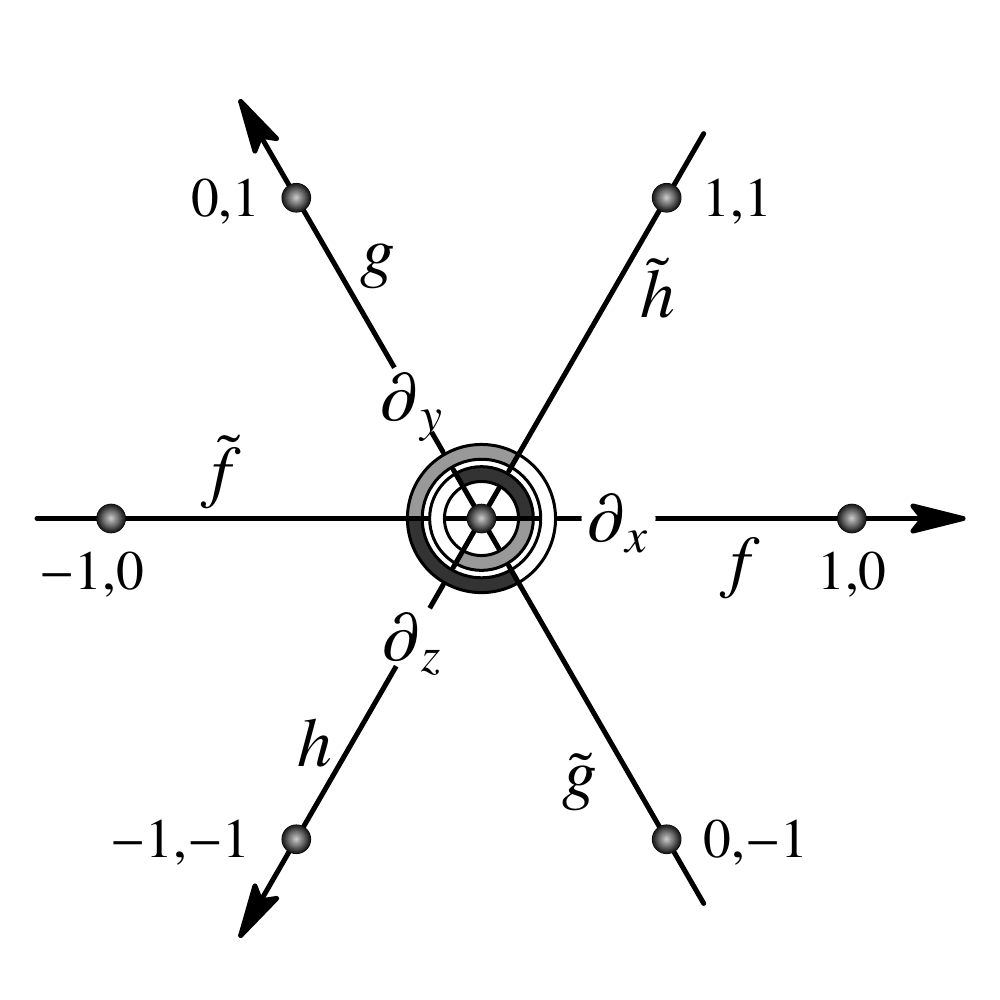}\qquad
\raisebox{15mm}{\includegraphics[scale=0.6]{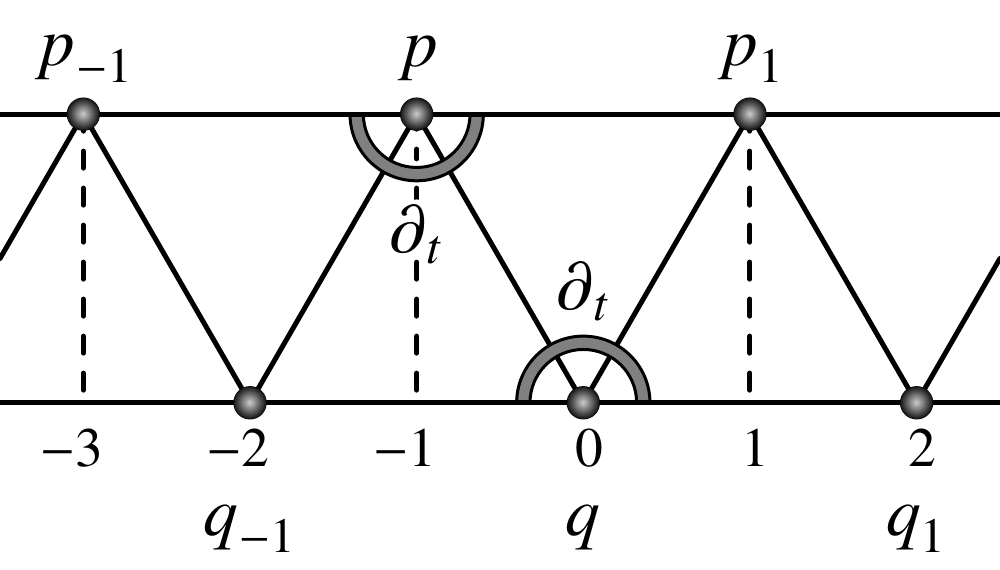}}}
\caption{Left: 7-point equation and differentiations $\partial_x$,
$\partial_y$, $\partial_z$.\hfil Right: embedding into triangle lattice of the second order lattice equation corresponding to the differentiation $\partial_t=\partial_y+\partial_z$.}
\label{fig:7-point}
\end{figure}

\begin{definition}\label{def:xyz}
We will say that the discrete equation (\ref{7-point}) admits a complete set of basic symmetries, if it is consistent, in the sense of Definition \ref{def:cons}, with each of three equations
\begin{alignat}{3}
\label{Dx}
 q_{,x}&= a(q, &&f-\tilde g-\tilde h)= a(q, &&\tilde f-g-h),\\
\label{Dy}
 q_{,y}&= b(q,-&&\tilde f+g-\tilde h)= b(q,-&&f+\tilde g-h),\\
\label{Dz}
 q_{,z}&= c(q,-&&\tilde f-\tilde g+h)= c(q,-&&f-g+\tilde h),
\end{alignat}
and the vector fields $\partial_x$, $\partial_y$, $\partial_z$ commute on equation (\ref{7-point}).
\end{definition}

Statement \ref{st:qpx} and $\Integer_3$-symmetry of the lattice imply that for any two of three equations (\ref{Dx}), (\ref{Dy}), (\ref{Dz}) there exist a common direction, one of three on the lattice, such that both these equation can be represented as 2-component system along this direction. For the sake of definiteness, let us consider the basic symmetries $\partial_y$ and $\partial_z$, then the variables $q(k)=q(k,j)$, $p(k)=q(k,j+1)$, $j=\const$, satisfy the following equations:
\begin{alignat}{2}
\nonumber
 q_{,y}&=b(q,-\tilde f(q,q_{-1})+g(q,p)-\tilde h(q,p_1)),\\
\label{qpy}
 p_{,y}&=b(p,-f(p,p_1)+\tilde g(p,q)-h(p,q_{-1})),\\[5pt]
\nonumber
 q_{,z}&=c(q,-f(q,q_1)-g(q,p)+\tilde h(q,p_1)),\\
\label{qpz}
 p_{,z}&=c(p,-\tilde f(p,p_{-1})-\tilde g(p,q)+h(p,q_{-1})).
\end{alignat}
This is the choice of variables which is most convenient when checking commutativity of the flows $\partial_y$ and $\partial_z$, because here we do not need to take  the discrete equation (\ref{7-point}) into account.

The next step of our construction is the change of variables
\begin{equation}\label{u-odd-even}
 u(2k)=q(k)=q(k,j),\quad u(2k-1)=p(k)=q(k,j+1),\quad j=\const,
\end{equation}
which can be described as projection of the vertices from two neighboring lines onto one common line (right fig.~\ref{fig:7-point}). As a result, systems (\ref{qpy}), (\ref{qpz}) turn into scalar lattice equations for the variable $u$, but with different equations for even and odd sites rather than autonomous ones:
\begin{alignat}{3}
\nonumber
 u_{,y}&=b(u,-\tilde f(u,u_{-2})+g(u,u_{-1})-\tilde h(u,u_1)),&& u=u(2k),\\
\label{uy}
 u_{,y}&=b(u,-f(u,u_2)+\tilde g(u,u_1)-h(u,u_{-1})),&& u=u(2k-1);\\[5pt]
\nonumber
 u_{,z}&=c(u,-f(u,u_2)-g(u,u_{-1})+\tilde h(u,u_1)),&& u=u(2k),\\
\label{uz}
 u_{,z}&=c(u,-\tilde f(u,u_{-2})-\tilde g(u,u_1)+h(u,u_{-1})),\quad && u=u(2k-1).
\end{alignat}
Finally, it remains to note that, under certain conditions, the equations for even and odd variables can be made identical by considering a linear combination of flows. This is reminiscent of the situation with the Ablowitz--Ladik lattice \cite{Ablowitz_Ladik_1975} which is the symmetrized sum of two commuting flows.

\begin{statement}\label{st:ut}
Let the equations under scrutiny are such that
\begin{equation}\label{bgg}
 b(u,v)=c(u,v),\quad g(u,v)=h(u,v),\quad \tilde g(u,v)=\tilde h(u,v),
\end{equation}
then the sum of the flows $\partial_t=\partial_y+\partial_z$ is described by autonomous second order evolution lattice equation:
\begin{multline}\label{ut}
 u_{,t}=  b(u,-\tilde f(u,u_{-2})+g(u,u_{-1})-\tilde g(u,u_1))\\
         +b(u,-f(u,u_2)-g(u,u_{-1})+\tilde g(u,u_1)).\qquad
\end{multline}
\end{statement}

It is clear that the choice of the pair of basic symmetries $\partial_z$, $\partial_x$ or $\partial_x$, $\partial_y$ instead of $\partial_y$, $\partial_z$ will bring to analogous formulas, up to the cyclic permutation of $a,b,c$ and $f,g,h$. The rest sections contain a number of examples that fall under the described scheme. Notice, that in some cases the symmetry condition (\ref{bgg}) can be relaxed.

\begin{figure}
\centerline{
\includegraphics[scale=0.6]{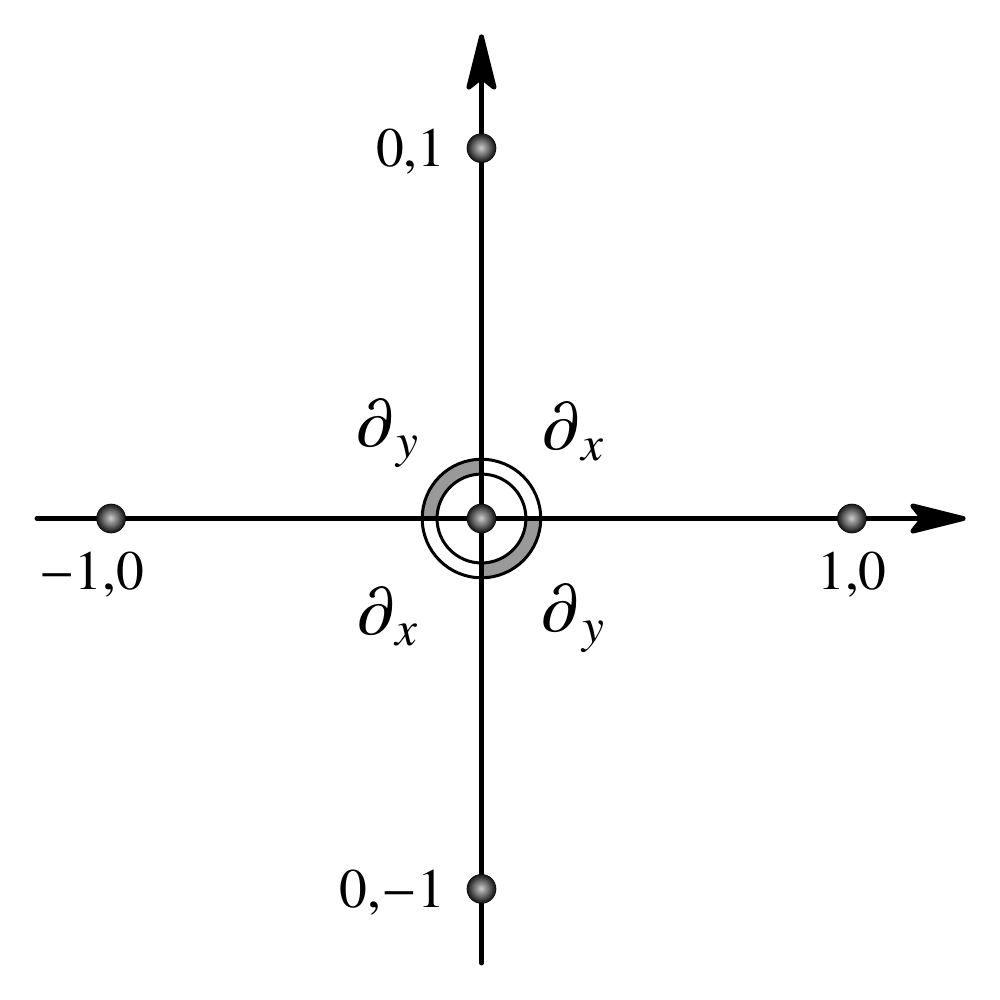}\quad
\raisebox{0mm}{\includegraphics[scale=0.6]{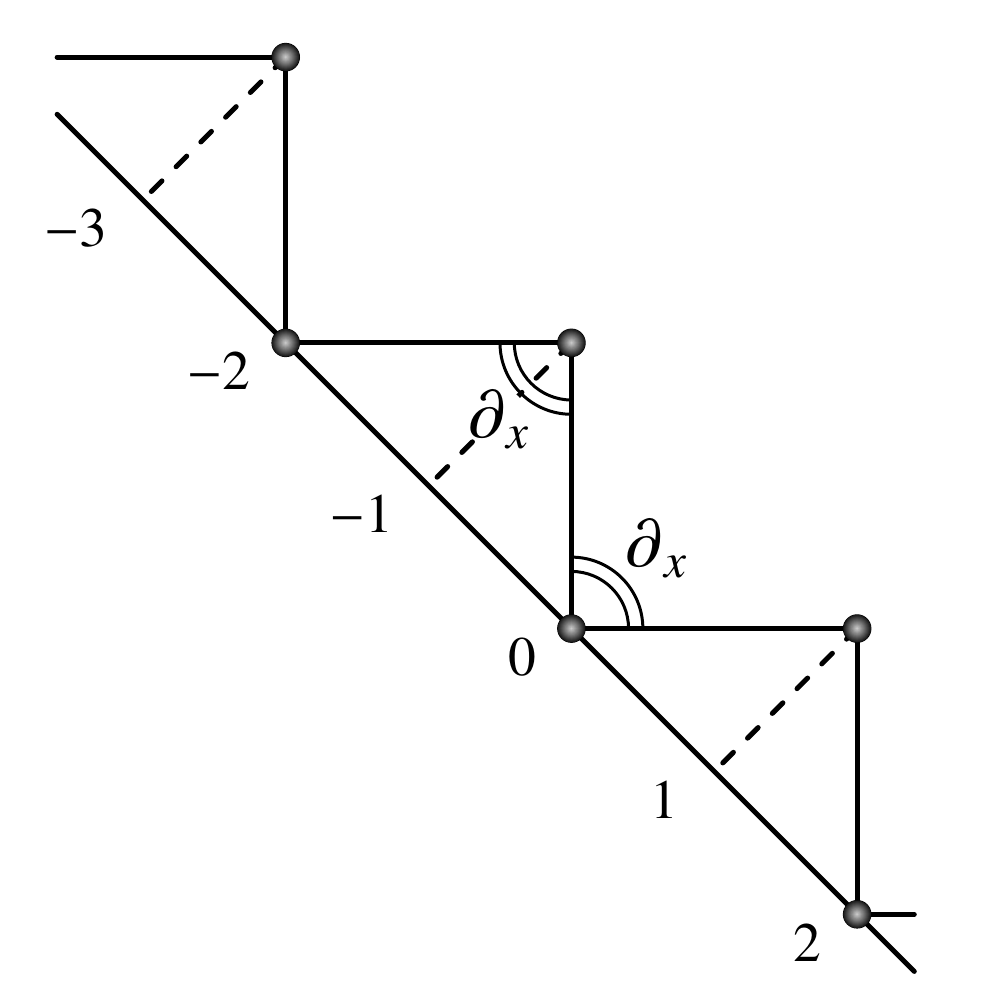}}}
\caption{5-point equation and embedding of the first order lattice equation into the square lattice.}\label{fig:5-point}
\end{figure}

For the sake of completeness, let us briefly consider the case $h=0$, when the basic symmetries $\partial_x$, $\partial_y$ coincide. After change of notation, we arrive, instead of 7-point equation, to 5-point one (see fig.~\ref{fig:5-point})
\begin{equation}\label{5-point}
 f(q,q_{1,0})-\tilde f(q,q_{-1,0})+g(q,q_{0,1})-\tilde g(q,q_{0,-1})=0
\end{equation}
and two basic symmetries
\begin{alignat}{1}
\label{5-Dx}
 q_{,x}&= a(q,f+g)= a(q,\tilde f+\tilde g),\\
\label{5-Dy}
 q_{,y}&= b(q,f-\tilde g)= b(q,\tilde f-g).
\end{alignat}
Each of equations (\ref{5-Dx}), (\ref{5-Dy}) can be written as a two-component lattice equation by choosing variables on a pair of neighboring lines parallel to any of the coordinate axes. Moreover, for equation (\ref{5-Dx}) it is possible to select variables on a pair of diagonals $n_1+n_2=j$, $n_1+n_2=j+1$, as it is shown on right fig.~\ref{fig:5-point} (analogously, for equation (\ref{5-Dy}) the lines $n_1-n_2=\const$ are suitable). The variables
\[
 u(2k)=q(k,j-k),\quad u(2k-1)=q(k,j+1-k),\quad j=\const,
\]
satisfy equations
\[
\begin{aligned}
 u_{,x}&=a(u,f(u,u_1)+g(u,u_{-1})), &\quad& u=u(2k),\\
 u_{,x}&=a(u,\tilde f(u,u_{-1})+\tilde g(u,u_1)), && u=u(2k-1).
\end{aligned}
\]
Obviously, in order for the equations at even and odd sites to coincide, it suffices to satisfy the symmetry conditions
\[
 f(u,v)=\tilde g(u,v),\quad \tilde f(u,v)=g(u,v).
\]

\begin{remark}
The compatibility condition (\ref{DxFF}) is rather restrictive with respect to the functions involved and, in principle, its analysis may result in a finite list of consistent pairs of equations (\ref{7-point}), (\ref{7-Dx}). At the moment, this problem is open, even for the special type of equations (\ref{d7-point}) considered in the next section. The classification of integrable lattice equations from Statement \ref{st:qpx} or, in more general setting, of the form
\[
 q_{,x}=a(q_{-1},q,p,p_1),\quad p_{,x}=b(q_{-1},q,p,p_1),
\]
is not known as well, although some examples are found in the literature. More results are known in the square lattice case, in particular, the problem of classifying 5-point equations (\ref{5-point}) admitting both basic symmetries (\ref{5-Dx}), (\ref{5-Dy}) is completely solved \cite{Adler_Shabat_2006}.
\end{remark}

\begin{remark}
In this paper we consider only autonomous equations (\ref{7-point}), although, in a more general setting, integrable equations of this type may contain variable parameters on the lattice (see eg.~\cite{Adler_Suris_2004}). This restriction is caused not only by reasons of simplicity, but also by the essence of the matter. In fact, one basic symmetry can exist also for an equation with variable parameters, but one can see by more detailed analysis of examples that the complete set of three basic symmetries is possible only in the autonomous case. It is not yet known whether the scheme given in this section admits any generalization to the non-autonomous case.
\end{remark}

\section{Shift-invariant examples}\label{s:difference}

\subsection{Basic symmetries}

Let us consider 7-point equations of the special form
\begin{equation}\label{d7-point}
 (T_1-1)(f)+(T_2-1)(g)+(T_3-1)(h)=0
\end{equation}
(recall that $T_3=T^{-1}_1T^{-1}_2$), where functions $f,g,h$ depend only on the differences of $q$ in the respective lattice direction:
\[
 f=f(q-q_{-1,0}),\quad g=g(q-q_{0,-1}),\quad h=h(q-q_{1,1}).
\]
This is the class of equations, which are invariant with respect to the one-parametric group of translations $q\to q+\varepsilon$ and are Lagrangian with respect to the functionals of the form
\[
 S=\sum_{i,j}T^i_1T^j_2\bigl(
 \varphi(q-q_{-1,0})+\psi(q-q_{0,-1})+\chi(q-q_{1,1})\bigr).
\]
A list of integrable equations (\ref{d7-point}) was obtained under assumption on existence of B\"acklund auto-transformation of special type \cite{Adler_2000a, Adler_2000b}. It turns out that all equations from this list also admit a complete set of basis symmetries (\ref{Dx}), (\ref{Dy}), (\ref{Dz}); moreover, functions $a,b,c$ in these equations do not depend explicitly on $q$. The basic symmetries for the case (\ref{A}) in the table \ref{t:fghabc} were found in \cite{Adler_2000a}; for the remaining cases this result seems to be new. We note that in the formulas for $f,g,h$ there are nonessential differences comparing to the cited papers, which are related to a more convenient choice of variables and, in some cases, rotation of the lattice.

\begin{table}[t]
\begin{alignat}{7}
\nonumber
 &f(x) &~~~~
 &g(x) &~~~~
 &h(x) &~~~~
 &a(x) &~~~~
 &b(x) &~~~~
 &c(x) &~~~~~~\\[3pt]
\hline\nonumber
 &&&&&&&&&&&&\\[-10pt]
\label{A}\tag{A}
  & \frac{\alpha}{x}
 && \frac{\beta}{x}
 && \frac{\gamma}{x}
 && \frac{1}{x}
 && \frac{1}{x}
 && \frac{1}{x} &\\[3pt]
\label{B}\tag{B}
  & \alpha\coth x
 && \beta\coth x
 && \gamma\coth x
 && \frac{1}{x}
 && \frac{1}{x}
 && \frac{1}{x} &\\[3pt]
\label{C}\tag{C}
  & \log\frac{x+\alpha}{x-\alpha}
 && \log\frac{x+\beta}{x-\beta}
 && \log\frac{x+\gamma}{x-\gamma}
 && \frac{1}{e^x-1}
 && \frac{1}{e^x-1}
 && \frac{1}{e^x-1} &\\[3pt]
\label{D}\tag{D}
  & -\log(1+1/x)
 && \log x
 && \log x
 && e^{-x}
 && \frac{1}{e^x+1}
 && \frac{1}{e^x+1} &\\[3pt]
\label{E}\tag{E}
 & \frac{1}{1+e^x}
 && -e^x
 && e^{-x}
 && x
 && \frac{1}{x}
 && \frac{1}{x+1} &\\[3pt]
\label{F}\tag{F}
  & \log(e^{-x}-1)
 && \log(e^x-1)
 && \log(e^x-1)
 && \frac{1}{e^x+1}
 && e^{-x}
 && e^{-x} &\\[3pt]
\label{G}\tag{G}
 & -x
 && -\log(e^{-x}-1)
 && \log(e^x-1)
 && \frac{1}{e^x+1}
 && e^{-x}
 && e^x &\\[3pt]
\label{H}\tag{H}
  & \log(e^{-x}-1)
 && \log\frac{e^x+1}{\gamma e^x+1}
 && \log(e^x+1)
 && \frac{1}{e^x+1}
 && e^{-x}
 && \frac{1}{e^x+\gamma} &\\[3pt]
\label{I}\tag{I}
  & \log\frac{\alpha e^x-1}{e^x-\alpha}
 && \log\frac{\beta e^x-1}{e^x-\beta}
 && \log\frac{\gamma e^x-1}{e^x-\gamma}
 && \frac{1}{e^x-1}
 && \frac{1}{e^x-1}
 && \frac{1}{e^x-1} & \\[5pt]
\hline\nonumber
\end{alignat}
\vspace{-30pt}
\captionsetup{width=0.8\textwidth}
\caption{Functions defining equation (\ref{d7-point}) and its basic symmetries. The parameters are constrained by relation $\alpha+\beta+\gamma=0$ in the cases (\ref{A}), (\ref{B}), (\ref{C}), and by relation $\alpha\beta\gamma=1$ in the case (\ref{I}).}
\label{t:fghabc}
\end{table}

\begin{statement}\label{st:d7-point}
Equation (\ref{d7-point}) is consistent with equations
\begin{alignat*}{3}
 q_{,x}&= a( &&T_1(f)-g-h)= a( &&f-T_2(g)-T_3(h)),\\
 q_{,y}&= b(-&&f+T_2(g)+h)= b(-&&T_1(f)+g-T_3(h)),\\
 q_{,z}&= c(-&&f-g+T_3(h))= c(-&&T_1(f)-T_2(g)+h),
\end{alignat*}
for all function sets $f,g,h,a,b,c$ from the table \ref{t:fghabc}. Moreover, the corresponding vector fields $\partial_x$, $\partial_y$, $\partial_z$ are mutually commutative and define the variational symmetries for the functional $S$.
\end{statement}

The statement \ref{st:d7-point} is verified by direct computation. As for the method that allows us to find the functions $a,b,c$, then it is based on the well-known connection of equations (\ref{7-point}) with quad-equations. Recall that if we consider quad-equations consistent on a cubic lattice, then the variables on the hyperplane $n_1+n_2+n_3=0$ are governed by equation of the form (\ref{7-point}). All three coordinate directions are transversal to this hyperplane, and the basic symmetries arise as a result of a continuous limit along one of these directions. A similar approach for derivation of continuous symmetries of 5-point equations was applied in \cite{Adler_Suris_2004}.

Equations in the table \ref{t:fghabc} correspond to type $H$ quad-equations and their asymmetric generalizations \cite{Adler_Bobenko_Suris_2009, Boll_Suris_2009}. It is rather tedious to consider their embedding in a 3-dimensional lattice, but in fact it is sufficient to know that as a result of the limiting procedure, one of the functions in the three-leg form of the quad-equation becomes a linear fractional function of $q_{,x}$. The coefficients of this function are easily found by the method of undetermined coefficients. In more details, let us notice that equations (\ref{d7-point}), for the above function sets $f,g,h$, can be cast into the rational form, after the change $e^q\to q$ when necessary. After this, we search functions $a,b,c$ in the linear fractional form like $a(x)=\frac{\kappa X+\lambda}{\mu X+\nu}$, where $X=x$ in the additive cases (\ref{A}), (\ref{B}), (\ref{E}) and $X=e^x$ in the multiplicative ones (all other cases, that is, such that functions $f,g,h$ contain logarithms). Moreover, since the symmetries are defined up to the change $q_{,t}\to c_1\partial_t(q)+c_2$ (scaling of $t$ and the classical translational symmetry), hence it is sufficient to consider $a=X$ and $a=1/(X+\nu)$. This trick works in all cases without exception.

\subsection{Lattice equations of second order in the symmetric cases}

Let us now turn to the construction of lattice equations (\ref{ut}). They can be attached to all equations from the above list under certain restrictions on the parameters. We begin with the cases when the conditions of discrete symmetry (\ref{bgg}) are fulfilled. For equations of the form (\ref{d7-point}), these conditions amount to relations
\begin{equation}\label{bg}
 b(x)=c(x),\quad g(x)=h(x).
\end{equation}
One can readily see from table \ref{t:fghabc}, that these equalities can be satisfied in all cases except for (\ref{E}), (\ref{G}). To this end, in the cases (\ref{A}), (\ref{B}), (\ref{C}) the parameters must be constrained by relations $\beta=\gamma$, $\alpha=-2\gamma$, moreover, it is possible to scale the parameter $\gamma$ into any non-vanishing constant. We will not consider the case (\ref{B}), because it is related with (\ref{A}) by point transformation $e^{2q}\to q$ \cite{Adler_2000a}. In the case (\ref{H}), one should set $\gamma=0$, and in the case (\ref{I}) the constraint is $\beta=\gamma$, $\alpha=\gamma^{-2}$, with essential parameter $\gamma$. All equations under consideration can be brought to a rational form by making, if necessary, the change $e^u\to u$. Taking into account also the scaling of $t$, we prove the following statement by direct computation.

\begin{statement}\label{st:dsym-ut}
Equations (\ref{ut}) corresponding to the cases (\ref{A}), (\ref{C}) at $\beta=\gamma$, $\alpha=-2\gamma$; (\ref{D}), (\ref{F}); (\ref{H}) at $\gamma=0$; (\ref{I}) at $\beta=\gamma$, $\alpha=\gamma^{-2}$, are, respectively, point equivalent to the following lattice equations:
\begin{align}
\label{utA}\tag{a}&
 \frac{1}{2}u_{,t}=
   \frac{1}{\frac{2}{u_2-u}+\frac{1}{u-u_1}+\frac{1}{u-u_{-1}}}
  -\frac{1}{\frac{2}{u_{-2}-u}+\frac{1}{u-u_1}+\frac{1}{u-u_{-1}}},\\[5pt]
\label{utC}\tag{c}&
 \frac{1}{4}u_{,t}=
   \frac{1}{
    \frac{(u_2-u+2)(u_1-u-1)(u_{-1}-u-1)}
         {(u_2-u-2)(u_1-u+1)(u_{-1}-u+1)}-1}
  +\frac{1}{
    \frac{(u_{-2}-u-2)(u_1-u+1)(u_{-1}-u+1)}
         {(u_{-2}-u+2)(u_1-u-1)(u_{-1}-u-1)}-1},\\[5pt]
\label{utD}\tag{d}&
 u_{,t}= \frac{u-u_{-1}}{u_1-u_{-1}+\frac{u-u_1}{u-u_2}}
        -\frac{u-u_1}{u_1-u_{-1}+\frac{u-u_{-1}}{u-u_{-2}}},\\[5pt]
\label{utF}\tag{f}&
 u_{,t}= u\frac{\bigl(1-\frac{u}{u_2}\bigr)\bigl(1-\frac{u_{-1}}{u}\bigr)}
               {1-\frac{u}{u_1}}
        +u\frac{\bigl(1-\frac{u_{-2}}{u}\bigr)\bigl(1-\frac{u}{u_1}\bigr)}
               {1-\frac{u_{-1}}{u}},\\[5pt]
\label{utH}\tag{h}&
 u_{,t}= u\frac{\bigl(1-\frac{u}{u_2}\bigr)\bigl(1+\frac{u_{-1}}{u}\bigr)}
               {1+\frac{u}{u_1}}
        +u\frac{\bigl(1-\frac{u_{-2}}{u}\bigr)\bigl(1+\frac{u}{u_1}\bigr)}
               {1+\frac{u_{-1}}{u}},\\[5pt]
\nonumber&
 -\frac{1}{1+\gamma^2}u_{,t}=
   \frac{(u-\gamma u_{-1})(u-\gamma u_1)(u-\gamma^{-2}u_2)}
        {(\gamma u_{-1}-u)(\gamma u_1-u_2)
        -(\gamma u-u_1)(u_{-1}-\gamma u_2)}\\
\label{utI}\tag{i}& \qquad\qquad\qquad
  -\frac{(u-\gamma^{-1}u_{-1})(u-\gamma^{-1}u_1)(u-\gamma^2u_{-2})}
        {(\gamma u_{-2}-u_{-1})(\gamma u-u_1)
        -(\gamma u_{-1}-u)(u_{-2}-\gamma u_1)}.
\end{align}
\end{statement}

Notice that equations (\ref{utF}), (\ref{utH}) coincide under the non-autonomous change $u(n)\to(-1)^nu(n)$.

The obtained equations look rather cumbersome and, as a rule, the bringing of the right-hand side to the common denominator only complicates the formulas. However, we should take into account the possibility of adding to the right-hand side terms of the form $c$ or $cu$ which correspond to the classical symmetries of translation and scaling. By construction, the lattice equation under consideration admits at least one of these symmetries. In all cases, except for (\ref{utF}), (\ref{utH}), it is not difficult to choose an appropriate linear combination which brings the equation to the factorized form (\ref{i.HPut})
\[
 u_{,t}=\frac{H(u_{-1},u)H(u,u_1)(u_2-u_{-2})}
          {P(u_{-2},u_{-1},u,u_1)P(u_{-1},u,u_1,u_2)},
\]
where $H$ is a biquadratic and $P$ is an affine-linear polynomials. This can be considered as a standard form; we will see in section \ref{s:Q4} that it exists also for more general lattice equations. Here we will not rewrite the equations in this form and give only a couple of examples. In the case (a), the right hand side does not require corrections and its factorization yields
\begin{equation}\label{utAA}
 H=(u_1-u)^2,\quad P=\frac{1}{2}(u_2+u)(u_1+u_{-1})-u_{-1}u_1-uu_2;
\end{equation}
in the case (i) the correction term is $(\gamma^2-\gamma^{-2})u$ and
\begin{equation}\label{utII}
 H=(\gamma+\gamma^{-1})(u_1-\gamma u)(\gamma u_1-u),\quad
 P=(\gamma u_2-u_{-1})(u_1-\gamma u)+(u_2-\gamma u_1)(\gamma u_{-1}-u).
\end{equation}

The remarkable fact is that all equations listed above are brought, by difference substitutions, to {\em the same} universal equation
\begin{equation}\label{wt}
 w_{,t}=(w+1)\Bigl(w_2\Bigl(\frac{1}{w_1}+1\Bigr)w
   -w\Bigl(\frac{1}{w_{-1}}+1\Bigr)w_{-2} +w_1-w_{-1}\Bigr).
\end{equation}
These substitutions are given in the following statement. The constant factors in the left-hand sides of the lattice equations were selected in such a way that an additional scaling of $t$ is not needed.

\begin{statement}\label{st:dsym-wt}
The lattice equations listed in Statement \ref{st:dsym-ut} are related with (\ref{wt}) by the following substitutions:
\begin{alignat*}{2}
(\ref{utA}):&\qquad&&
 w=\frac{1}{2\frac{(u_3-u_2)(u_1-u)}{(u_3-u_1)(u_2-u)}-1},\\
(\ref{utC}):&&&
 w=\frac{(u_3-2u_2+u_1)(u_2-2u_1+u)}
        {2-2(u_2-u_1)^2-(u_3-2u_2+u_1)(u_2-2u_1+u)},\\
(\ref{utD}):&&&
 w=-\frac{1}{1+\frac{u_2-u_1}{(u_3-u_1)(u_2-u)}},\\
(\ref{utF}):&&&
 w=\frac{u}{u_1}-1,\\
(\ref{utH}):&&&
 w=-\frac{u}{u_1}-1,\\
(\ref{utI}):&&&
 w=\frac{\gamma(u-(\gamma^{-1}+\gamma)u_1+u_2)
         (u_1-(\gamma^{-1}+\gamma)u_2+u_3)}
    {(\gamma u-u_1)(\gamma u_2-u_3)-(\gamma u_1-u_2)(u-\gamma u_3)}.
\end{alignat*}
\end{statement}

\begin{remark}
All these substitutions are of the form $w=\phi(\rho[u])$, where $\rho[u]$ is equivalent to the density $\log\partial_2(f)$ of the conservation law for the corresponding lattice equation $u_{,t}=f[u]$. Equations (\ref{wt}), (\ref{utA}) and the substitutions between them were considered in paper \cite{Adler_2016}, devoted to integrable lattice equations which are invariant under the group of M\"obius transformations $u\to\frac{\kappa u+\lambda}{\mu u+\nu}$. It is easy to verify that equation (\ref{utA}) and substitution $u\to w$ can be written in terms of the invariants $X$ and $Y/u_{,t}$ of this group,
\[
 X=\frac{(u_1-u)(u_{-1}-u_{-2})}{(u_1-u_{-1})(u-u_{-2})},\quad
 Y=\frac{(u_1-u)(u-u_{-1})}{u_1-u_{-1}},
\]
namely,
\[
 u_{,t}=\frac{4Y(1-X-X_1)}{(2X-1)(2X_1-1)},\quad w=\frac{1}{2X-1}.
\]
It is interesting that there is also another substitution into (\ref{wt}) in this case; it is of the form $w=H/P-1$ with polynomials (\ref{utAA}).
\end{remark}

\subsection{Asymmetric cases}

In the rest cases (\ref{G}), (\ref{E}), we first apply the change $e^q\to q$. In the case (\ref{G}) the basic symmetries $\partial_y$ and $\partial_z$ take the form
\begin{alignat*}{2}
 q_{,y}&=q_{-1,0}\Bigl(\frac{q}{q_{0,1}}-1\Bigr)
                 \Bigl(\frac{q}{q_{1,1}}-1\Bigr)
       &=\frac{1}{q_{1,0}}(q-q_{0,-1})(q-q_{-1,-1}),\\
 q_{,z}&=q_{1,0}\Bigl(\frac{q}{q_{0,1}}-1\Bigr)
                \Bigl(\frac{q}{q_{1,1}}-1\Bigr)
       &=\frac{1}{q_{-1,0}}(q-q_{0,-1})(q-q_{-1,-1}).
\end{alignat*}
One can see that the change (\ref{u-odd-even}) used in the symmetric examples,
\[
 u(2k)=q(k)=q(k,j),\quad u(2k-1)=p(k)=q(k,j+1),\quad j=\const,
\]
brings to equations which are different for odd and even variables $u$, for any linear combination of the flows $\partial_y$ and $\partial_z$. However, this can be easily fixed by applying the change
\[
 u(2k)=q(k)=q(k,j),\quad u(2k-1)=p(k)=-\frac{1}{q(k,j+1)},\quad j=\const
\]
instead. As a result, the basic symmetries turn into the following pair of 2-component lattice equations \cite{Tsuchida_2002}:
\begin{gather*}
 q_{,y}=q_{-1}(qp+1)(qp_1+1),\quad p_{,y}=-p_1(qp+1)(q_{-1}p+1),\\
 q_{,z}=q_1(qp+1)(qp_1+1),\quad p_{,z}=-p_{-1}(qp+1)(q_{-1}p+1),
\end{gather*}
and the difference of the flows $\partial_t=\partial_z-\partial_y$ will be equivalent to the lattice \cite{Garifullin_Yamilov_2012, Garifullin_Mikhailov_Yamilov_2014}
\begin{equation}\label{utG}\tag{g}
 u_{,t}=(u_1u+1)(uu_{-1}+1)(u_2-u_{-2}).
\end{equation}
Like usually, this equation admits a substitution into equation (\ref{wt}), $w=uu_1$ \cite{Adler_2016}. A remarkable feature of equation (\ref{utG}) is that this is the only example with a polynomial right-hand side from all the equations arising in our approach.

In the case (\ref{E}), let us apply a similar change
\[
 u(2k)=q(k)=\frac{1}{q(k,j)},\quad u(2k-1)=p(k)=q(k,j+1),\quad j=\const.
\]
As the result, the flows $\partial_y$, $\partial_z$ take the form
\begin{gather*}
 q_{,y}= \frac{q+q_{-1}}{1+(q+q_{-1})(p_1+p)},\quad
 p_{,y}=-\frac{p_1+p}{1+(q+q_{-1})(p_1+p)},\\
 q_{,z}=-\frac{q_1+q}{1+(q_1+q)(p_1+p)},\quad
 p_{,z}= \frac{p+p_{-1}}{1+(q+q_{-1})(p+p_{-1})},
\end{gather*}
and the flow $\partial_t=-\partial_z-\partial_y$ is equivalent to the lattice equation
\begin{equation}\label{utE}\tag{e}
u_{,t}=\frac{u_2-u_{-2}}{(1+(u_2+u)(u_1+u_{-1}))(1+(u_1+u_{-1})(u+u_{-2}))}.
\end{equation}
Substitution into equation (\ref{wt}) is of the form
\[
 w=\frac{1}{1+(u_3+u_1)(u_2+u)}-1.
\]

\section{Elliptic lattice equation}\label{s:Q4}

In \cite{Adler_Suris_2004} equations (\ref{7-point}) in the multiplicative form
\begin{equation}\label{Q}
 f(q,q_{1,0};\alpha)
 f(q,q_{-1,0};\alpha)
 f(q,q_{0,1};\beta)
 f(q,q_{0,-1};\beta)
 f(q,q_{-1,-1};\gamma)
 f(q,q_{1,1};\gamma)=1,
\end{equation}
were considered, with $\alpha+\beta+\gamma=0$ and factors $f$ of the general form
\[
 f(q,p;\alpha)=\frac{s(q+p+\alpha)s(q-p+\alpha)}
                    {s(q+p-\alpha)s(q-p-\alpha)}.
\]
Equations of this type are related with quad-equation $Q_4$ which is associated with elliptic curve and its trigonometric and rational degenerations $Q_3|_{\delta=1}$ and $Q_2$ (moreover, cases (\ref{A}), (\ref{C}), (\ref{I}) from the previous section can be obtained by further limiting procedures, corresponding to equations $Q_3|_{\delta=0}$ and $Q_1$). Equation (\ref{Q}) is equivalent to the Euler--Lagrange equation for the functional
\[
 S=\sum_{i,j}T^i_1T^j_2\bigl(
 \varphi(q,q_{-1,0};\alpha)+\varphi(q,q_{0,-1};\beta)
 +\varphi(q,q_{1,1};\gamma)\bigr),\quad
 \partial_q\varphi(q,p;\alpha)=\log f(q,p;\alpha),
\]
provided that $\varphi(q,p;\alpha)=\varphi(p,q;\alpha)$, which is true if $s$ is an odd function.

\begin{statement}\label{st:Q}
Equation (\ref{Q}) corresponding to functions
\[
 s(x)=\sigma(x),\quad s(x)=\sinh(x),\quad s(x)=x,
\]
where $\sigma(x)$ is the Weierstrass function, is consistent with equations
\begin{align*}
 \frac{q_{,x}+1}{q_{,x}-1}
  &=f(q,q_{1,0};\alpha)f(q,q_{0,-1};\beta)f(q,q_{1,1};\gamma)\\
  &=f(q,q_{-1,0};-\alpha)f(q,q_{0,1};-\beta)f(q,q_{-1,-1};-\gamma),\\[3pt]
 \frac{q_{,y}+1}{q_{,y}-1}
  &=f(q,q_{-1,0};\alpha)f(q,q_{0,1};\beta)f(q,q_{1,1};\gamma)\\
  &=f(q,q_{1,0};-\alpha)f(q,q_{0,-1};-\beta)f(q,q_{-1,-1};-\gamma),\\[3pt]
 \frac{q_{,z}+1}{q_{,z}-1}
  &=f(q,q_{-1,0};\alpha)f(q,q_{0,-1};\beta)f(q,q_{-1,-1};\gamma)\\
  &=f(q,q_{1,0};-\alpha)f(q,q_{0,1};-\beta)f(q,q_{1,1};-\gamma).
\end{align*}
The basic flows $\partial_x$, $\partial_y$, $\partial_z$ mutually commute and define the variational symmetries for the functional $S$.
\end{statement}

The proof is based, like for the Statement \ref{st:d7-point}, on embedding of (\ref{Q}) into the cubic lattice and continuous limit in the corresponding quad-equations (cf with equation (9.3) in \cite{Adler_Suris_2004} for the 5-point case).

The symmetry condition (\ref{bgg}) for equations under consideration amounts to the equality of the parameters $\beta=\gamma$. The resulting lattice equation is the following.

\begin{statement}\label{st:Q-ut}
Let $\alpha=-2\gamma$, $\beta=\gamma$, then the flow $\partial_t=\frac{1}{2}(\partial_y+\partial_z)$ is governed by the lattice equation
\begin{multline}\label{Qut}
 u_{,t}=\frac{1}{f(u,u_2;2\gamma)f(u,u_1;-\gamma)f(u,u_{-1};-\gamma)-1}\\
  -\frac{1}{f(u,u_1;-\gamma)f(u,u_{-1};-\gamma)f(u,u_{-2};2\gamma)-1}\qquad
\end{multline}
for the variables $u(2k)=q(k,j)$, $u(2k-1)=q(k,j+1)$.
\end{statement}

After replacing of $f$ with its expression through the function $s$, equation (\ref{Qut}) becomes very bulky, but it can be simplified by use of the point transformation of the form
\[
 \phi(u)=\tilde u,\quad \kappa(\gamma)=\tilde\gamma,\quad
 \partial_t=c(\gamma)\partial_{\tilde t}.
\]
Function $\phi$ is chosen in such a way that the underlying quad-equation is brought to the affine-linear form \cite{Adler_Suris_2004}, in this case equation (\ref{Q}), the basic flows $\partial_x,\partial_y,\partial_z$ and the lattice equation (\ref{Qut}) are brought to rational form with numerator and denominator of minimal degree.

\begin{theorem}\label{th:HPut}
Equation (\ref{Qut}) corresponding to functions $s(x)=\sigma(x)$, $s(x)=\sinh(x)$, $s(x)=x$, is point equivalent to an equation of the form
\begin{equation}\label{HPut}
 u_{,t}=\frac{H(u_{-1},u)H(u,u_1)(u_2-u_{-2})}
  {P(u_{-2},u_{-1},u,u_1)P(u_{-1},u,u_1,u_2)},
\end{equation}
where $P$ is an affine-linear polynomial obeying the symmetry of a square with the vertices $u_{-1},u,u_1,u_2$:
\begin{equation}\label{HP1}
 P(u_{-1},u,u_1,u_2)=P(u_2,u_1,u,u_{-1})=P(u_1,u,u_{-1},u_2),
\end{equation}
and $H(u,u_1)=H(u_1,u)$ is a symmetric biquadratic polynomial related with $P$ by equality
\begin{equation}\label{HP2}
 \const H(u,u_1)=\partial_{-1}(P)\partial_2(P)-\partial_{-1}\partial_2(P)P.
\end{equation}
The family of equations (\ref{HPut}) is invariant with respect to the group of the M\"obius transformations $u\to\frac{\kappa u+\lambda}{\mu u+\nu}$. The polynomial $P$ corresponding to the case $s(x)=\sigma(x)$ belongs to the generic orbit of this group.
\end{theorem}

\begin{proof}
In the case $s=\sigma(x)$, we apply the change
\[
 \wp(u)=\tilde u,\quad \wp(\gamma)=\tilde\gamma,
 \quad\partial_t=4\wp'(\gamma)^2\wp'(2\gamma)\partial_{\tilde t}.
\]
In order to pass from $\sigma$ to $\wp$ Weierstrass functions, we use the identity
\[
 -\frac{\sigma(a+b)\sigma(a-b)}{\sigma(a)^2\sigma(b)^2}=\wp(a)-\wp(b).
\]
Replace here $a$ with $a+b$ and with $a-b$, then taking the quotient brings to the following identity as a corollary:
\[
 \frac{\sigma(a+2b)\sigma(a-b)^2}{\sigma(a-2b)\sigma(a+b)^2}
  =\frac{\wp(a+b)-\wp(b)}{\wp(a-b)-\wp(b)}.
\]
Function $f(q,p;\alpha)$ is rewritten in the form
\[
 f(q,p;\alpha)=\frac{\sigma(q+\alpha)^2}{\sigma(q-\alpha)^2}\cdot
  \frac{\wp(q+\alpha)-\wp(p)}{\wp(q-\alpha)-\wp(p)}
\]
by use of the first identity; then the second one makes possible to replace the $\sigma$-functions in the common factor of the product of three functions $f$:
\begin{multline*}
 f(u,u_2;2\gamma)f(u,u_1;-\gamma)f(u,u_{-1};-\gamma)\\
 =\frac{(\wp(u+\gamma)-\wp(\gamma))^2}
       {(\wp(u-\gamma)-\wp(\gamma))^2}
  \cdot\frac{\wp(u+2\gamma)-\wp(u_2)}{\wp(u-2\gamma)-\wp(u_2)}
  \cdot\frac{\wp(u-\gamma)-\wp(u_1)}{\wp(u+\gamma)-\wp(u_1)}
  \cdot\frac{\wp(u-\gamma)-\wp(u_{-1})}{\wp(u+\gamma)-\wp(u_{-1})}.
\end{multline*}
This expression is already rational with respect to the variables $\tilde u_2$, $\tilde u_{\pm1}$. Now, we only have to expand the instances of the form $\wp(u+\const)$ by use of the relations
\[
 \wp'(a)^2=4\wp(a)^3-g_2\wp(a)-g_3,\quad
 \wp(a+b)+\wp(a)+\wp(b)=
 \frac{1}{4}\left(\frac{\wp'(a)-\wp'(b)}{\wp(a)-\wp(b)}\right)^2,
\]
and to substitute this into equation (\ref{Qut}), omitting the tilde at new variables $u,t$ and parameter $\gamma$. Rather tedious, but direct computation brings to the lattice equation of the form (\ref{HPut}) with the polynomials
\begin{equation}\label{Qut-ell}
\begin{aligned}
 H&= \left(uu_1+\gamma u+\gamma u_1+\frac{g_2}{4}\right)^2
     -(u+u_1+\gamma)(4\gamma uu_1-g_3),\\
 P&= (12\gamma^2-g_2)\left(u_{-1}uu_1u_2+\frac{g^2_2}{16}\right)\\
  &\quad -(4\gamma^3+g_2\gamma+2g_3)
          (u_{-1}uu_1+u_{-1}uu_2+u_{-1}u_1u_2+uu_1u_2-g_3)\\
  &\quad     +2\left(\gamma^4-g_3\gamma-\frac{g^2_2}{16}\right)(u_{-1}+u_1)(u+u_2)\\
  &\quad -\gamma(4\gamma^3-3g_2\gamma-4g_3)(u_{-1}u_1+uu_2)\\
  &\quad +\gamma\left(g_2\gamma^2+6g_3\gamma+\frac{g^2_2}{4}\right)
          (u_{-1}+u+u_1+u_2),
\end{aligned}
\end{equation}
which satisfy the relations (\ref{HP1}), (\ref{HP2}). This form of polynomials $P,H$ is related only with the special gauge associated with the Weierstrass form of the elliptic curve (a simpler parametrization is provided by the Jacobi form). The invariance of equation (\ref{HPut}) with respect to the M\"obius transformations is easily verified. The general form of polynomials $P$ with the required type of symmetry is
\begin{multline*}
 P= a_4u_{-1}uu_1u_2+a_3(u_{-1}uu_1+u_{-1}uu_2+u_{-1}u_1u_2+uu_1u_2)\\
  +a_2(u_{-1}+u_1)(u+u_2)+b_2(u_{-1}u_1+uu_2)+a_1(u_{-1}+u+u_1+u_2)+a_0.
\end{multline*}
The procedure of reconstruction of affine-linear polynomial by the respective biquadratic one \cite{Adler_Bobenko_Suris_2009} proves that the polynomials (\ref{Qut-ell}) correspond to the generic orbit, indeed.

Analogously, in the trigonometric case $s(x)=\sinh(x)$ one finds
\begin{equation}\label{Qut-trig}
\begin{gathered}
 \cosh(2u)=\tilde u,\quad e^{2\gamma}=\tilde\gamma,\quad
 \partial_t=-\sinh(4\gamma)\partial_{\tilde t},\\
 H=(\gamma+\gamma^{-1})(u_1-\gamma u)(\gamma u_1-u)
    +\frac{1}{4}(1-\gamma^2)^2(1+\gamma^{-2}),\\
 P=(\gamma u_2-u_{-1})(u_1-\gamma u)+(u_2-\gamma u_1)(\gamma u_{-1}-u)
    +\frac{1}{4}(1-\gamma^2)^2(1+\gamma^{-2})
\end{gathered}
\end{equation}
(again, we omit tilde at $u,\gamma$ in the formulas for $H,P$); in the rational case $s(x)=x$ the change is given by equations
\begin{equation}\label{Qut-rat}
\begin{gathered}
 u^2=\tilde u,\quad \partial_t=-16\gamma\partial_{\tilde t},\\
 H=(u_1-u)^2-2\gamma^2(u_1+u)+\gamma^4,\\
 P=\frac{1}{2}(u_2+u)(u_1-u_{-1})-u_2u-u_1u_{-1}-
   \gamma^2(u_{-1}+u+u_1+u_2)+3\gamma^4.
\end{gathered}
\end{equation}
Computations in these case are straightforward.
\end{proof}

\begin{remark}
The polynomial $P$ (\ref{Qut-ell}) differs from the polynomial $Q$ from the respective quad-equation $Q_4$ by the dependence on parameter: $P$ contains only one parameter $\gamma$, associated with the square itself, while $Q$ depends on parameters $\gamma,\gamma'$ associated with pairs of opposite edges of the square. The relation between two polynomials is described by the limiting procedure
\[
 \const P(\gamma)=
 \lim_{\gamma'\to\gamma}\frac{Q(\gamma,-\gamma')}{\gamma-\gamma'}.
\]
The polynomials (\ref{Qut-trig}) and (\ref{Qut-rat}) are related in similar way with quad-equations $Q_3|_{\delta=1}$ and $Q_2$.
\end{remark}

The following property is proven by direct computation.

\begin{statement}
Equations (\ref{HPut}) and (\ref{wt}) are related by substitution $w=H/P-1$ both in the trigonometric (\ref{Qut-trig}) and the rational (\ref{Qut-rat}) cases.
\end{statement}

In the elliptic case, such a substitution does not pass; more generally, in this case I was not able to find any substitution to another second-order lattice equation.

Finally, let us mention the degenerations leading to the lattice equations from section \ref{s:difference}. In the limiting case $\gamma=0$, polynomials (\ref{Qut-rat}) coincide with polynomials (\ref{utAA}) which correspond to the lattice equation (\ref{utA}) (notice, that substitution $\gamma=0$ directly into (\ref{Qut}) makes no sense because $f(q,p;0)\equiv1$; one has to apply scaling to $t$ at first). If $\gamma\ne0$ then one can set $\gamma=1$ wthout loss of generality, by scaling of $u$ and $t$. In the case (\ref{Qut-trig}), the scaling $u\to\delta^{-1}u$ and passing to the limit $\delta\to0$ bring to polynomials (\ref{utII}) corresponding to equation (\ref{utI}).

\phantomsection{}
\addcontentsline{toc}{section}{Conclusion}
\section*{Conclusion}

In this paper, we studied evolution second order lattice equations associated with special continuous symmetries for discrete 7-point equations on a triangular lattice. Except for a few of the most degenerate cases (equations (\ref{utF}), (\ref{utG}), (\ref{utH})), all these equations belong to the family (\ref{HPut}). Except for the maximally non-degenerate case, all lattice equations admit a difference substitution to equation (\ref{wt}). The proposed construction scheme leaves no doubt about the integrability of all obtained examples (a convincing indication is that the shift along the two-dimensional triangular lattice can be interpreted as a B\"acklund transformation). However, such aspects of integrability as higher symmetries, zero curvature representations, soliton solutions, are still unexplored and need additional investigation, possibly individual for each lattice equation. It is of interest to study also more general two-component systems connected with the basic continuous symmetries of 7-point equations.

\phantomsection{}
\addcontentsline{toc}{section}{Acknowledgements}
\subsection*{Acknowledgements}

This work was supported by the RFBR grant \# 16-01-00289a.

\phantomsection{}
\addcontentsline{toc}{section}{Refernces}

\end{document}